\documentclass[submission, copyright]{eptcs}
\providecommand{\event}{PLACES 2023}

\usepackage[T1]{fontenc}
\usepackage{underscore}
\usepackage{amsmath}
\usepackage{amsthm}
\usepackage{amssymb}
\usepackage{mathrsfs}
\usepackage{fancyvrb}
\usepackage{color}
\usepackage{tikz}
\usetikzlibrary{automata, positioning}
\usepackage{enumitem}
\usepackage{bussproofs}
\usepackage{listings}

\theoremstyle{plain}
\newtheorem{thm}{Theorem}[section]

\newtheorem{prop}[thm]{Proposition}

\theoremstyle{definition}
\newtheorem{definition}{Definition}[section]

\newtheorem{example}{Example}[section]
\theoremstyle{remark}
\newtheorem*{remark}{Remark}

\newcommand{\ecfsm}{CAA}

\title{Communicating Actor Automata -- Modelling Erlang Processes as Communicating Machines}
\author{%
Dominic Orchard\institute{University of Kent, UK}\email{d.a.orchard@kent.ac.uk}\and%
Mihail Munteanu\institute{Masabi Ltd.}\email{mihailmunteanu944@gmail.com} \and%
Paulo Torrens\institute{University of Kent, UK}\email{paulotorrens@gnu.org}%
}
\date{\today}

\newcommand{\titlerunning}{Communicating Actor Automata}
\newcommand{\authorrunning}{D. Orchard, M. Munteanu \& P. Torrens}

\hypersetup{
  bookmarksnumbered,
  pdftitle    = {\titlerunning},
  pdfauthor   = {\authorrunning},
  pdfsubject  = {EPTCS},               % Consider adding a more appropriate subject or description
}

\begin{document}

\maketitle

\begin{abstract}
Brand and Zafiropulo’s notion of Communicating Finite-State Machines (CFSMs) provides \linebreak a succinct and powerful model of message-passing concurrency, based around channels. However, a major variant of message-passing concurrency is not readily captured by CFSMs: the actor model. In this work, we define a variant of CFSMs, called Communicating Actor Automata, to capture the actor model of concurrency as provided by Erlang: with mailboxes, from which messages are received according to repeated application of pattern matching. Furthermore, this variant of CFSMs supports dynamic process topologies, capturing common programming idioms in the context of actor-based message-passing concurrency. This gives a new basis for modelling, specifying, and verifying Erlang programs. We also consider a class of CAAs that give rise to freedom from race conditions.

\end{abstract}

\section{Introduction}
\label{sec:introduction}
Modern software development often deviates from the traditional approach of sequential computation and thrives on \emph{concurrency}, where code is written such that several processes may run simultaneously while potentially sharing resources. Out of the increasing complexity of software systems, platforms and programming languages were created to facilitate the development of concurrent, parallel, and distributed computations, such as the Erlang programming language~\cite{armstrong1997development,armstrong2007history}. The need for formal specification of such systems has motivated the design of formal systems that allow programs to be reasoned about.

The Communicating Finite-State Machines (CFSMs) of Brand and Zafiropulo (also known as \emph{communicating automata}) provide a model for describing concurrent, communicating processes in which a notion of \emph{well-formed} communication protocols can be described~\cite{brand1983communicating}. The essential idea is to model a (finite) set of concurrent processes as finite automata (one automaton per process) whose labels correspond to sending and receiving messages on channels connecting each pair of automata, collectively called a \emph{protocol}. Through such precise descriptions, properties of communicating processes can be studied, e.g., checking whether every message is received, or whether no process is left awaiting for a message which is never sent.

This model is useful for further studying the decidability, or undecidability, of various properties of concurrent systems (e.g.,~\cite{pachl2003reachability,peng1992analysis,finkel2017synchronizability,gouda1984closed,gouda1984progress}). For example, Brand and Zafiropulo show that boundedness (i.e., that communication can proceed with bounded queues), deadlock freedom, and unspecified reception (sending messages that aren't received) are all decidable properties when restricted to two machines with a single type of message~\cite{brand1983communicating}. Furthermore, deciding these properties can be computed in polynomial time~\cite{rosier1984deciding}, and can be computed when only one machine is restricted to a single type of message. CFSMs have also been employed more recently as a core modelling tool that provides a useful interface between other models of concurrent programs, such as graphical choreographies~\cite{DBLP:conf/popl/LangeTY15}.

In the CFSM model, processes communicate via channels (FIFO queues) linking each process. Therefore a process knows from which other process a message is received. This differs to the actor model where each process has a `mailbox' into which other processes deposit messages, not necessarily with any information about the sender. This makes it difficult to capture actor-based approaches in traditional CFSMs. A further limitation of CFSMs is that they capture programs with fixed communication topologies: both sender and receiver of a message are fixed in the model, and a process cannot have its messages dynamically targeted to different processes. However, this is not the predominant programming idiom in concurrent programming settings. CFSMs also prescribe simple models of message reception, and do not capture more fine-grained reception methods, such as Erlang's mailbox semantics which allow processing messages other than the most recent one, leveraging pattern matching at the language level. Even so, CFSMs are tantalisingly close to Erlang's computational model, with every pair of processes representing sequential computation that is able to communicate bidirectionally.

We propose a variant of CFSMs to capture Erlang's asynchronous
mailbox semantics, and furthermore allow dynamic topologies through
the binding (and rebinding) of variables for process identifiers
via a notion of memory within a process' automaton. Section~\ref{sec:caa-definition} explicates the model including examples of
models corresponding to simple Erlang programs. We consider properties
of such models in Section~\ref{sec:compatibility}. Section~\ref{sec:discussion}
concludes with a discussion, some related work, and next steps.

\subsection{Communicating Finite-State Machines}

To facilitate comparison, we briefly recap the formal definition of
CFSMs~\cite{brand1983communicating}. A system of $N$-CFSMs is referred
to as a \emph{protocol} which comprises four components (usually
represented as a $4$-tuple):
\begin{itemize}
\item $(S_i)^{N}_{i = 1}$ are $N$ (disjoint) finite sets $S_i$
giving the set of states of each process $i$;

\item $(o_i)^N_{i = 1}$ where $o_i \in S_i$
are the initial states of each process $i$;

\item $(M_{i,j})^N_{i,j=1}$ are $N^2$ (disjoint) finite sets where
$M_{i,j}$ represents the set of messages that can be sent from process $i$
to process $j$, and where $M_{i,j} = \emptyset$ when $i = j$;

\item $(\mathsf{succ} : (S_i \times \bigcup^N_{j=1}(M_{i,j} \uplus M_{j,i}))
\rightarrow S_i)^N_{i = 1}$ are $N$
state transition functions (\emph{partial} functions) where
$\mathsf{succ}(s, l)$ computes the successor state $s'$ for process
$i$ from the state $s$ and given a message $l$ that is either being
sent from $i$ to $j$ (thus $l \in M_{i,j}$) or being received from $j$
by $i$ (thus $l \in M_{j,i}$).
\end{itemize}
The typical presentation views the above indexed sets as finite
sequences. We use the notation $l$ for messages as we later
refer to these as being `labels' of automata.

For a protocol, a \emph{global state} (or configuration) is a pair
$(S, C)$ of a sequence of states for each process $S = \langle{s_1 \in
S_1, \ldots s_N \in S_N}\rangle$ and $C$ is an $N \times N$ matrix
whose elements $c_{i,j} \in C$ are finite sequences drawn from
$M_{i,j}$ representing a FIFO queue (channel) of messages between
process $i$ and $j$ (here and later we use $\cdot$ to concatenate
such sequences and $[l_1, \ldots, l_m]$ for an instance
of a sequence with $m$ messages).

A binary-relation \emph{step} captures when one global state
$(S, C)$ can evolve into another global state $(S', C')$ due to a
single $\mathsf{succ}$ function. That is, $(S', C')
\in \mathsf{step}(S, C)$ iff there exists $i$, $k$, $l_{i,k}$ and
 either:
\begin{align*}
& \begin{array}{rll}
1. & \text{ ($i$ sends to $k$):} & s'_i \in S' = \mathsf{succ}_i(s_i, l_{i,k})
\ \text{and}\ c'_{i,k} \in C' = c_{i,k} \cdot [l_{i,k}] \\
\textbf{or}\;\; 2. & \text{ (reception from $k$ by $i$):} & s'_k \in S' = \mathsf{succ}_k(s_k, l_{i,k})\ \text{and}\
c_{i,k} \in C = [ l_{i,k} ] \cdot c'_{i, k}
\end{array}
\end{align*}
We adopt similar terminology and structure, but
vary enough to capture  the actor model of Erlang.

\section{A variant of CFSMs for Erlang}
\label{sec:caa-definition}
Our main goal is to define a CFSM variant, which we call Communicating Actor Automata (CAA), by borrowing from Erlang's mailbox semantics. We first review some core Erlang concepts, and follow by describing CAA and their composition into protocols. While the traditional definition of CFSMs immediately considers a global configuration (state) of some processes, we take each process (representing an Erlang actor) as a separate state machine, which gives a local ``in-isolation'' characterisation from which the global ``protocol'' characterisation is derived.

\subsection{Erlang basic definitions}
\label{subsec:erlang-basics}

\newcommand{\varX}{X}
\newcommand{\varY}{Y}
\newcommand{\tmT}{t}
\newcommand{\tmP}{p}
\newcommand{\tmQ}{q}
\newcommand{\patP}{pat}

A key concept in Erlang is that of a \emph{mailbox}: instead of processing
messages strictly in FIFO order, each process (also referred to as an
\emph{actor}) possesses a queue of incoming messages from which they may
\emph{match}: given a sequence of patterns, a process picks the first message
from the queue that unifies with one of those patterns, in an ordered fashion.
If no pattern matches the first message in the queue, the next message is tried
for all patterns, and so on~\cite{carlsson2001introduction}. In the following,
we will use Erlang's term syntax in an abstract way (for details, see Carlsson
et al.~\cite{carlsson2000core}), whose actual choice may vary. Concrete syntax is given in \texttt{monospace}
font, e.g., $\texttt{\{1, 2\}}$ is an Erlang tuple of two integer terms.

\begin{definition}[Syntactic categories]
  Let $\textit{Term}$ be the set of terms, ranged over by $\tmT$. Let $\textit{Var} \subset \textit{Term}$ be the set of variables, ranged over by uppercase letters. Let $\mathit{Proc} \subset \textit{Var}$ be the set of process identifiers which uniquely identify processes. Let $\mathit{Pat}$ be the set of patterns, ranged over by $\patP$. % Subscripts and superscripts are employed to distinguish multiple elements.
  As Erlang has a call-by-value semantics, we define a subset $\textit{Value}\subset\textit{Term}$ of terms which
  are normal forms (called \emph{values}), ranged over by $v$. Notice that $\textit{Var}\subset\textit{Value}$.
\end{definition}

We note that, in regard to the semantics proposed in this paper, we mostly focus
on four basic kinds of terms: process identifiers, variables, atoms and tuples. While
identifiers and variables allow us to control the topology, atoms and tuples are useful for
structuring message patterns (aiding in identifying which message is to be sent
or received). In the example that will be given in Section
\ref{subsec:multi}, we shall also consider integers and arithmetic operators,
but that's not necessary. It would also be possible to consider functions in the
term syntax, but we won't entertain this possibility in this paper as we intend to
focus on the expressivity of the process automaton itself and not on the term
language.

\newcommand{\unify}{\rhd}

In order to mimic Erlang's method of receiving messages, we need a notion of
unification: incoming messages are matched against a set of patterns, and will
proceed only if one of those is accepted.

\begin{definition}[Unification]
We define the notion of \emph{unification}
\cite{carlsson2001introduction,carlsson2000core} between a term and a
pattern written $\tmT \unify \patP$ which either yields $\bot$
representing failure to pattern match or it yields an
\emph{environment} which is a finite map from a subset of variables $V
\subseteq \textit{Var}$ to terms, i.e., $\Gamma_V : V \rightarrow
\textit{Term}$ represents the binding context of a successful pattern
match. Such maps are ranged over by $\gamma$. If $t \unify pat = \gamma$, with $\gamma \in \Gamma_V$ for some $V$, then $pat$ becomes equal to $t$ if
we replace every variable $v \in V$ in it by $\gamma_V(v)$. We write
$\{X_1\mapsto t_1,\ ...,\ X_n\mapsto t_n\}$ to denote the environment
that maps $X_i$ to $t_i$ (for all $1 \leq i \leq n$).
\end{definition}

%\begin{example} The following gives a few examples of (successful and
%  unsuccessful) unifications:
%
%\begin{align*}
%\begin{array}{ll}
%\texttt{sym} \unify \texttt{sym} & =\enspace \{\} \\
%\texttt{\{msg, 1\}} \unify \texttt{\{msg, X\}} & =\enspace \{X \mapsto
%  \texttt{1}\} \\
%\texttt{[1, 2, 3]} \unify \texttt{[H | T]} & =\enspace \{H \mapsto \texttt{1}, T \mapsto
%\texttt{[2, 3]}\} \\
%\texttt{[1, 2, 3]} \unify \texttt{[X | [X | Y]]} & =\enspace \bot \\
%\texttt{[1, 1, 3]} \unify \texttt{[X | [X | Y]]} & =\enspace \{X \mapsto
%  1,Y \mapsto \texttt{[3]}\}
%\end{array}
%\end{align*}
%\end{example}

% The actual definition of unification is not relevant to the rest of this article (thus it will remain abstract), though more details can be found in the definition of the \emph{Core Erlang} language~\cite{carlsson2001introduction,carlsson2000core}.

%\begin{definition}[Process identifier environments]
%We define a specialised notion of environment just for process identifiers, written $\Gamma_V$, which maps from a subset of Erlang variables $V \subseteq \mathit{Var}$ to process identifiers, i.e., $\Gamma_V : V \rightarrow \mathit{Proc}$. Elements of this set are ranged over by $\gamma \in \Gamma_V$.
%\end{definition}

\subsection{\ecfsm{}s for individual processes}
\label{subsec:single}

Just as in CFSMs, we describe actors by finite-state automata. Each automaton is
described individually and represents a single Erlang process, having its own
unique identifier. Our main intention is to capture possible states in an Erlang
process, by saying which messages it's allowed to receive and to send at a given
point during execution, and what it should do after it.

We formally define our notion of actor as follows (note
the addition of final states, not always included for CFSMs).

\begin{definition}
  A Communicating Actor Automata (\ecfsm{}) is a 7-tuple $(S,\ o,\ F,\ 
  \mathcal{L}_!,\ \mathcal{L}_?,\ \delta, <)$, which includes a finite set of
  states $S$, an initial state $o \in S$, a possibly empty set of final states
  $F \subseteq S$, a set of send labels $\mathcal{L}_! \subseteq \mathit{Term}$,
  a finite set of receive labels $\mathcal{L}_? \subseteq \mathit{Pat}$, a
  function $\delta : S \times (\mathcal{L}_? \cup (\mathit{Var} \times
  \mathcal{L}_!)) \rightarrow S$, \linebreak describing transitions,
  and an $S$-indexed family of order relations $<$ on $\mathcal{L}_?$. We impose a
  restriction on the domain of $\delta$ such that a state may only have either
  some number of receive labels or a single send label (but never both).
  We write $\delta(s, ?\mathit{pat})$
  for any transitions which receive a message matching $\mathit{pat}$, and $\delta(s, X!t)$ for
  transitions which send a message, where $X!t$ denotes the pair $(X, t)$ of a
  message given by term $t$ being sent to process $X$. For an state $s \in S$, we
  write $<_s$ as the order relation on such state.
\end{definition}
%
%can remove if we need to:
%
%Note that we add final states here whereas CFSMs do not always include
%this component. We find this useful for characterising deadlock.

Notice a \ecfsm{} is essentially a deterministic finite automata (DFA) state machine with the alphabet defined as $\Sigma = \mathcal{L}_? \cup (\mathit{Var} \times \mathcal{L}_!)$. Non-determinism is exposed by interaction of many \ecfsm{}s in a protocol, which will be described in Section~\ref{subsec:multi}. As such, a notion of a ``static \ecfsm{}'' (as opposed to mobile) can be conceived of as having transitions $\delta : S \times (\mathcal{L}_? \cup (\mathit{Proc} \times \mathcal{L}_!)) \rightarrow S$ where we replace the variables associated with send labels by concrete process identifiers, i.e., the target of a send is always known ahead of time. We do not explore this notion further here.

\begin{example}
\label{exm:mem}
Consider the following Erlang code which defines a function \texttt{mem} which emulates a memory cell via recursion:
\begin{lstlisting}[language=Erlang,basicstyle=\ttfamily\small,keywordstyle=\color{blue}\bf,deletekeywords={get,put}]
mem(S) -> receive
            {get, P} -> P!S, mem(S);
            {put, X} -> mem(X)
          end.
\end{lstlisting}
The state of the memory cell is given by variable \texttt{S}. The process receives either a pair \texttt{\{get, P\}}, after which it sends \texttt{S} to the process identifier \texttt{P}, or a pair \texttt{\{put, X\}}, after which it recurses with \texttt{X} as the argument (the `updated' memory cell state). Note that lowercase terms in Erlang are atoms.

Once spawned, this function can be modelled as a \ecfsm{} with states $S = \{s_0,\ s_1\}$, initial and final states $o = s_0$ and $F = \{s_0\}$, send labels $\mathcal{L}_! = \mathit{Terms}$, receive labels $\mathcal{L}_? = \{\texttt{\{get, P\}},\ \texttt{\{put,
    S\}}\}$, order stating $\texttt{\{get, P\}} <_{s_0} \texttt{\{put, S\}}$, and the following transition (with corresponding automaton):

\vspace{0.5em}

\begin{minipage}{0.5\textwidth}
\centering
\begin{align*}
\begin{array}{ll}
\delta(s_0,\ ?\texttt{\{get, P\}}) & = s_1 \\
\delta(s_0,\ ?\texttt{\{put, S\}}) & = s_0 \\
\delta(s_1,\ \texttt{P} ! \texttt{S}) & = s_0
\end{array}
\end{align*}
\end{minipage}{\scalebox{0.9}{\begin{minipage}{0.5\textwidth}
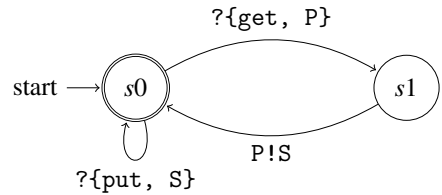
\centering
\begin{tikzpicture}[shorten >=1pt,node distance=1.7cm and 4cm,on grid,auto]
    \node[state,initial,accepting] (s0)   {$s0$};
    \node[state] (s1) [right=of s0] {$s1$};

    \path[->]
    (s0) edge [bend left] node {\texttt{?\{get, P\}}} (s1)
    (s1) edge [bend left] node {\texttt{P!S}} (s0)
    (s0) edge [loop below] node {\texttt{?\{put, S\}}} ();
\end{tikzpicture}
\end{minipage}}}
\end{example}

\vspace{0.5em}

\newcommand{\sequ}[1]{[#1]}
\newcommand{\sequEmpty}{\epsilon}
\noindent
Notice we don't use the variable $X$ in the above example for the `put' message, as we want to rebind the received value (in the second component of the pair) to $S$ in the recursive call. We do not consider the additional aspect of Erlang's semantics in which already bound variables may appear in pattern matches, incurring a unification, which is a further complication not considered in this paper.

The above definition is enough to capture the static semantics of an actor. However, during execution, further information is needed to represent its dynamic behaviour at runtime: namely the mailbox and the internal state of the actor. We proceed to formally define this.

\begin{definition}
A \emph{local state} (or \emph{machine configuration}) is a triple $(s, m, \gamma)$, being comprised of a state $s$, a finite sequence of terms $m$, and an environment $\gamma$. The sequence of terms $m$ models \emph{message queues}, also called actor \emph{mailboxes}, of unreceived messages, and $\gamma$ is the actor's memory. We write $\sequEmpty$ for the empty sequence and $\sequ{t_1, \ldots, t_n}$ for the sequence comprising
$n$ elements with $t_1$ being the head of the
queue. Two sequences $m,\ m'$ can be appended, written $m \cdot m'$, e.g.,
$\sequ{1,2,3} \cdot \sequ{4,5} = \sequ{1,2,3,4,5}$.
\end{definition}

\subsection{Systems of \ecfsm{}s: protocols,
states, and traces}
\label{subsec:multi}

As with a CFSM model or an Erlang program, computation is described by the communication among concurrent actors.
In order to formally define that, we give an operational semantics to the combination of several CAAs, called a protocol,
through an evaluation step relation between states.

\begin{definition}[Protocol]
A \emph{protocol} is an indexed family of \ecfsm{}s, of finite size (or arity) $N$, written
$\langle{(S_i,\ o_i,\ F_i,\ \mathcal{L}_{!i},\ \mathcal{L}_{?i},\ \delta_i)}\rangle^N_{i = 1}$. Each index $i$ represents the unique process identifier of each process.
\end{definition}

\begin{definition}[Global state]
A \emph{global state} (or \emph{system configuration}) for
a protocol comprises a finite sequence of $N$
local states, written $\langle{(s_i, m_i, \gamma_i)}\rangle^N_{i = 1}$
where every $s_i \in S_i$. We denote the set of global states for a protocol with arity $N$ as $G_N$.
\end{definition}

Given a protocol, we may derive what we call an initial global state: before starting,
each process has an empty mailbox, and is in its initial state as defined in its
own CAA. This initial state is deterministic: for any given protocol, there's a
single possible initial state.

\begin{definition}[Initial global state]
For a protocol, the \emph{initial global state}
is $\langle{(o_i,\ \sequEmpty,\ \emptyset)}\rangle^N_{i = 1}
\in G_N$,
i.e., every machine is in its initial state with an empty mailbox,
and its mapping from variables to process identifiers is empty.
\end{definition}

In order to use the mailbox semantics, we define a partial function that is defined
only if a message may be accepted at the moment. As in Erlang, we look for each message
that's in the mailbox in order, and only try the next message if no currently
accepting pattern matches. While checking each message, patterns are checked in
the defined order and only the first one that matches will be accepted.

\begin{definition}[Pick function]
We implictly assume a CAA $(S,\ o,\ F,\ \mathcal{L}_!,\ \mathcal{L}_?,\ \delta,\ <)$ as our context. Then,
the partial function $pick(s, m, v) = \langle pat, \gamma\rangle$ is defined if and only if:
\begin{itemize}
    \item For all $pat' \in \mathcal{L}_?$, if $\delta(s,?pat)$ is defined, then for
    all $v_k \in m$, we have that $v_k \unify pat = \bot$;
    \item For all $pat' <_s pat$, we have $v \unify pat' = \bot;$ and
    \item There is a $\gamma$ such that $v \unify pat = \gamma$.
\end{itemize}
\end{definition}

Finally, we now can define our semantics through a step relation, which defines what happens
to a global state once one of the current possible transitions is performed.

\begin{definition}[Step relation]
Given a protocol $\langle{(S_i,\ o_i,\ F_i,\ \mathcal{L}_{!i},\ \mathcal{L}_{?i},\
    \delta_i, <_i)}\rangle^N_{i = 1}$,
the relation\footnote{Notice $\mathit{step}$
can be equivalently
thought of as a non-deterministic function
(producing many possible global states)
or as a relation from a single input global state to
many outcome global states. The definition
here declaratively defines this relation.} \emph{step} \linebreak denotes the non-deterministic
transitions of the overall concurrent system,
of type $\mathit{step} : G_N \rightarrow \mathcal{P}(G_N)$. We write $\mathit{step}(c_1) \ni c_2$ if $c_2$ is amongst the possible outcomes of $\mathit{step}(c_1)$, as defined by the following two
rules:

\begin{prooftree}
\RightLabel{(SEND)}
\AxiomC{$\delta_i(s_i, P!t) = s'_i$}
\AxiomC{$\gamma_i(P) = j$}
\AxiomC{$\gamma_i(t) \rightarrow^* v$}
\TrinaryInfC{$\mathit{step}(\langle{(s_i, m_i, \gamma_i)}\rangle^N_{i = 1}) \quad\ni\quad \langle{...,
(s'_i, m_i, \gamma_i), ...,
(s_j, m_j \cdot \sequ{v}, \gamma_j), ...}\rangle$}
\end{prooftree}

\vspace{0.25em}

\begin{prooftree}
\RightLabel{(RECV)}
\AxiomC{$\delta_j(s_j, ?\mathit{pat}) = s'_j$}
\AxiomC{$m_j = m \cdot \sequ{v} \cdot m'$}
\AxiomC{$pick(s_j, m, v) = \langle pat, \gamma\rangle$}
\TrinaryInfC{$\mathit{step}(\langle{(s_i, m_i, \gamma_i)}\rangle^N_{i = 1}) \quad\ni\quad \langle{...,
(s'_j, m \cdot m', \gamma_j \cup \gamma), ...}\rangle$}
\end{prooftree}
\end{definition}

\noindent
The first rule, (SEND), actions a \emph{send} transition with label $P!t$
for a process $i$ in state $s_i$ resulting
in the state $s'_i$. We lookup the
variable $P$ from the process identifier environment
$\gamma_i$ to get the process identifier
we are sending to, i.e., $\gamma_i(P) = j$, which means we
are sending to process $j$. Notice that as $\gamma_i$ is a finite map, $\gamma_i(P) = P$ if $P$ is not in the domain of $\gamma_i$. By abuse of notation, $\gamma(t)$ replaces occurrences of variables of $\gamma$ in $t$,
allowing for the use of internal state, and then we evaluate the resulting term to a value $v$ (which is denoted by the reduction relation $\rightarrow^*$). Subsequently in the resulting
global configuration, process $i$ is now in state $s'_i$
and process $j$ has the message $v$ enqueued onto the
end of its mailbox in its configuration, while the states for any other processes is kept the same.

The second rule, (RECV), actions a \emph{receive} transition with
label $\mathit{pat}$ for a process $j$ in state $s_j$
resulting in the process $j$ being in state $s'_j$
under the condition that the mailbox $m_j$ contains
a message $v$ at some point with prefix $m$ and suffix
$m'$. We use the $pick$ function to check the semantic constraints: we want to
pick the first message $v$ for which some possible pattern matches, and the first
one that does so. If this is the correct pattern, given state, prefix and value, then the value term $v$ will unify with $\mathit{pat}$ to produce
an binding $\gamma$.  Subsequently, process $j$ is now in state
$s'_j$ with $v$ removed from its mailbox and its process
identifier environment updated to $\gamma_j \cup \gamma$: its internal state gets updated by any terms matched while receiving the message. Any states for processes other than $j$ remains the same.

Following that, a way to describe a possible result for computation is through a trace, which represents a sequence of steps from the initial global state into one possible final state (as \emph{step} is non-deterministic, several outcomes are possible). We proceed by formally defining a notion of trace.

\begin{definition}[Trace]
A \emph{trace} is a sequence of system configurations such that the first one is the initial global state, and
the subsequent states are obtained from applying the \emph{step} relation onto the previous configuration.
We call a sequence of global states $T$ a trace if it has the form $T = \langle{t_0, t_1, ...,
    t_x}\rangle$ and satisfies the following conditions:
\begin{itemize}
\setlength{\itemsep}{1pt}
\setlength{\parskip}{0pt}
\setlength{\parsep}{0pt}
\item $t_0 = \langle{(o_i, \sequEmpty, \emptyset)}\rangle^N_{i = 1}$
\item  $\mathit{step}(t_i) \ni t_{i+1}$
\item $\mathit{step}(t_x) = \emptyset$ where $t_x$ is the last term in the trace sequence.
\end{itemize}
\end{definition}

%\begin{definition}[Run]
%  A \emph{run} for a protocol is a sequence of global states starting from
%  the initial step and progressing such that each successive global
%  state is attained from the \textit{step} relation.
%\end{definition}

\begin{example}
    \label{exm:trace}
    Recall the \texttt{mem} function from Example~\ref{exm:mem}, to
    which we assign process id $0$.
    We consider a protocol comprises four machines, with three further machines
    with process identifiers $1$, $2$ and $3$, given by the following definitions:
    {\small{
    \begin{align*}
    \setlength{\arraycolsep}{0.5em}
    \begin{array}{lllll}
        S_1 = \{c_0, c_1, c_2, c_3\} &
        o_1 = c_0 &
        F_1 = \{c_3\} &
        \mathcal{L}_{!1} = \{\texttt{\{get, 1\}}, \texttt{\{put, X+1\}}\} &
        \mathcal{L}_{?1} = \{\texttt{X}\}
    \end{array} \\
    \begin{array}{lll}
        \delta_1(c_0, 0!\texttt{\{get, 1\}}) = c_1 &
        \delta_1(c_1, ?\texttt{X}) = c_2 &
        \delta_1(c_2, 0!\texttt{\{put, X+1}\}) = c_3
        \end{array} \\[1em]
        \begin{array}{lllll}
        S_2 = \{d_0, d_1, d_2, d_3\} &
        o_2 = d_0 &
        F_2 = \{d_3\} &
        \mathcal{L}_{!2} = \{\texttt{\{get, 2\}}, \texttt{\{put, X+2\}}\} &
        \mathcal{L}_{?2} = \{\texttt{X}\}
    \end{array} \\
        \begin{array}{lll}
        \delta_2(d_0, 0!\texttt{\{get, 2\}}) = d_1 &
        \delta_2(d_1, ?\texttt{X}) = d_2 &
        \delta_2(d_2, 0!\texttt{\{put, X+2}\}) = d_3
        \end{array} \\[1em]
        \begin{array}{llllll}
        S_3 = \{e_0, e_1\} &
        o_3 = e_0 &
        F_3 = \{e_1\} &
        \mathcal{L}_{!3} = \{\texttt{\{put, 0\}}\} &
        \mathcal{L}_{?3} = \emptyset
        \end{array} \\
        \begin{array}{lll}
          \delta_3(e_0, 0!\texttt{\{put, 0\}}) = e_1 &
        \end{array}
    \end{align*}
    }}

    \noindent
    The first machine above (process id $1$) makes a `get' request to process $0$ then receives
    the result \texttt{X} and sends back to $0$ a `put' message with \texttt{X+1}.
    The second machine above (process id $2$) is similar to the first, requesting the value from
    process $0$ but then sending back a `put' message with \texttt{X+2}. The third
    machine (process id $3$) sends to $0$ a `put' message with the initial value $0$.

We can then get the following trace for the protocol
$\langle{(S_i, o_i, F_i, \mathcal{L}_{!i}, \mathcal{L}_{?i},
  \delta_i)}\rangle^3_{i = 1}$, one of the many possibilities, demonstrating mobility and the ability to store internal state (we underline the parts of the configuration which have changed at
each step of the trace for clarity):

\pagebreak

\begin{align*}
  \begin{array}{rllll}
    & \langle(s_0, \sequEmpty, \emptyset), & (c_0, \sequEmpty, \emptyset), & (d_0, \sequEmpty, \emptyset), & (e_0, \sequEmpty, \emptyset)\rangle \\
    & \langle\underline{(s_0, \sequ{\{\texttt{put}, 0\}}, \emptyset)}, & (c_0, \sequEmpty, \emptyset), & (d_0, \sequEmpty, \emptyset), & \underline{(e_1, \sequEmpty, \emptyset)}\rangle \\
    & \langle\underline{(s_0, \sequEmpty, \{S \mapsto 0\})}, & (c_0, \sequEmpty, \emptyset), & (d_0, \sequEmpty, \emptyset), & (e_1, \sequEmpty, \emptyset)\rangle \\
    & \underline{\langle(s_0, \sequ{\{\texttt{get}, 1\}}, \{S \mapsto 0\})}, & \underline{(c_1, \sequEmpty, \emptyset)}, & (d_0, \sequEmpty, \emptyset), & (e_1, \sequEmpty, \emptyset)\rangle \\
    & \langle\underline{(s_1, \sequEmpty, \{S \mapsto 0, P \mapsto 1\})}, & (c_1, \sequEmpty, \emptyset), & (d_0, \sequEmpty, \emptyset), & (e_1, \sequEmpty, \emptyset)\rangle \\
    & \langle\underline{(s_0, \sequEmpty, \{S \mapsto 0, P \mapsto 1\})}, & \underline{(c_1, \sequ{0}, \emptyset)}, & (d_0, \sequEmpty, \emptyset), & (e_1, \sequEmpty, \emptyset)\rangle \\
    & \langle(s_0, \sequEmpty, \{S \mapsto 0, P \mapsto 1\}), & \underline{(c_2, \sequEmpty, \{X \mapsto 0\})}, & (d_0, \sequEmpty, \emptyset), & (e_1, \sequEmpty, \emptyset)\rangle \\
    & \langle\underline{(s_0, \sequ{\{\texttt{put}, 1\}}, \{S \mapsto 0, P \mapsto 1\})}, & \underline{(c_3, \sequEmpty, \{X \mapsto 0\})}, & (d_0, \sequEmpty, \emptyset), & (e_1, \sequEmpty, \emptyset)\rangle \\
    & \langle\underline{(s_0, \sequ{\{\texttt{put}, 1\}, \{\texttt{get}, 2\}}, \{S \mapsto 0, P \mapsto 1\})}, & (c_3, \sequEmpty, \{X \mapsto 0\}), & \underline{(d_1, \sequEmpty, \emptyset)}, & (e_1, \sequEmpty, \emptyset)\rangle \\
    & \langle\underline{(s_0, \sequ{\{\texttt{get}, 2\}}, \{S \mapsto 1, P \mapsto 1\})}, & (c_3, \sequEmpty, \{X \mapsto 0\}), & (d_1, \sequEmpty, \emptyset), & (e_1, \sequEmpty, \emptyset)\rangle \\
    & \langle\underline{(s_1, \sequEmpty, \{S \mapsto 1, P \mapsto 2\})}, & (c_3, \sequEmpty, \{X \mapsto 0\}), & (d_1, \sequEmpty, \emptyset), & (e_1, \sequEmpty, \emptyset)\rangle \\
    & \langle\underline{(s_0, \sequEmpty, \{S \mapsto 1, P \mapsto 2\})}, & (c_3, \sequEmpty, \{X \mapsto 0\}), & \underline{(d_1, \sequ{1}, \emptyset)}, & (e_1, \sequEmpty, \emptyset)\rangle \\
    & \langle(s_0, \sequEmpty, \{S \mapsto 1, P \mapsto 2\}), & (c_3, \sequEmpty, \{X \mapsto 0\}), & \underline{(d_2, \sequEmpty, \{X \mapsto 1\})}, & (e_1, \sequEmpty, \emptyset)\rangle \\
    & \langle\underline{(s_0, \sequ{\{\texttt{put}, 3\}}, \{S \mapsto 1, P \mapsto 2\})}, & (c_3, \sequEmpty, \{X \mapsto 0\}), & \underline{(d_3, \sequEmpty, \{X \mapsto 1\})}, & (e_1, \sequEmpty, \emptyset)\rangle \\
    & \langle\underline{(s_0, \sequEmpty, \{S \mapsto 3, P \mapsto 2\})}, & (c_3, \sequEmpty, \{X \mapsto 0\}), & (d_3, \sequEmpty, \{X \mapsto 1\}), & (e_1, \sequEmpty, \emptyset)\rangle
    \end{array}
\end{align*}
\end{example}

\section{Characterising CAA systems: race freedom
and convergence}
\label{sec:compatibility}
We consider how to characterise race conditions between CAAs and identify
a subclass of CAA systems which is race free, or exhibiting \emph{convergence}.
To characterise race conditions, we first
define the notion of what possible messages can be observed
as `incoming' to a process at a particular step in a trace. We define this
notion via multisets of messages:

\begin{definition}[Multiset of messages in a mailbox]
  Given a mailbox $m$, we denote by $A_m$ the multiset of elements in $m$, i.e.,
  there is a way of `arranging' the elements of $A_m$ to obtain $m$.
\end{definition}

\begin{definition}[Incoming messages multiset]
Let $G_N$ = $\langle{c_1,...,(s_i,m_i,\gamma_i),...,c_N}\rangle$
be a global configuration with process $i$ at state $s_i$, where the next step
of the system gives $M$ possible configurations:
\begin{align*}
  \textit{step}(G_N) = \{\langle{c_1{_1}, ..., (s_i{_1},m_i{_1},\gamma_i{_1}), ..., c_N{_1}}\rangle, ...
  \langle{c_1{_M}, ..., (s_i{_M},m_i{_M},\gamma_i{_M}), ..., c_N{_M}}\rangle\}
\end{align*}
The \emph{incoming message multiset} $I_i$ represents all the possible incoming messages
for process $i$ defined:
\begin{align*}
   I_i = (\bigcup\limits_{y=1}^{M}A_{m_{iy}}) - A_{m_i}
\end{align*}
i.e., we take the union of all the possible mailbox multisets $A_{m_{iy}}$ for $i$ obtained
after a step is taken, from which is
subtracted (multiset difference) the messages in the mailbox of $i$ before that step. That is,
the incoming messages $I_i$ are all possible messages that can be added to the mailbox of $i$
after taking a step.
\end{definition}

The definition of $I$ is implicitly parameterised by the starting
configuration $G_N$. We will typically superscript $I$ to denote it
occurring at some position $x$ in a trace, i.e., $I^x_i$.

\begin{remark}
We use overloaded notation $m_i = \langle{I^0_i,I^1_i, ..., I^n_i}\rangle$ to represent a mailbox where the first message is drawn from $I^0_i$, second from $I^1_i$ and so on. That is, $m_i = \langle{l_{0i},l_{1i}, ..., l_{ni}}\rangle$ where $l_{xi} \in I^x_i$.

The mailbox of a process $i$ can now be written as a sequence of multisets of incoming messages, more precisely a subsequence of $\langle{I^0_i,I^1_i, ..., I^n_i}\rangle$. Thus, $m_i = \langle{I^{x_0}_i,I^{x_1}_i, ..., I^{x_m}_i}\rangle$, where $x_0 < ... < x_m$. But why a subsequence? If message $l_{xi}$ is consumed, then multiset $I^x_i$ disappears from the mailbox.
\end{remark}

\begin{remark}[Transition function $\delta$ notation overloading] Usually a transition function
  $\delta$ has as its argument a pair of a state and a label. We overload the second part
  of this pair to allow a multiset such that $\delta(s,?I^x_i) = \{s' \vert , \forall l \in I^x_i, \delta (s,?l) = s'\}$, i.e., the set of target states for any of the possible messages in $I^x_i$.
\end{remark}

\begin{definition}[Race condition]
\label{def:race-condition}
For a protocol, a \emph{race condition} represents the scenario in which there exists a state that can consume 2 or more messages from the same $I^x_i$ and cannot consume any messages from the previous mailbox sets.
That is:
\begin{align*}
& \forall y \in [0,x-1] . \;\; \vert \delta(s,?I^y_i) \vert = 0
\; \wedge \; \vert \delta(s,?I^x_i) \vert \geq 2
\end{align*}
Intuitively, if in the current state for a process we can receive 2 or more messages from set $I^x_i$ we are faced with a race condition, since these messages were sent at the same step and could arrive in any order. This is represented by the second part of the above conjunction. However, in order for us to reach $I^x_i$, we need to not consume any messages before that, hence the first part of the conjunction. If a previous multiset of messages has just one message we can consume, the race condition will not take place.

\end{definition}

%Note that this represents a race condition since there are other
% possible runs resulting in different \emph{put} messages arriving
% to the first machine.

\paragraph{Trace convergence and race freedom for systems of two automata}

We consider a class of binary CAA systems (i.e., where $N
= 2$) that is race free by showing that its traces always converge, that
is, the system is deterministic. Furthermore, this makes the testing
of a two automata system much easier, as we need only examine one
trace to determine the final state of the system.

We recall that our definition of CAAs allows only for a restricted set of
transitions: if a state has a transition, it can only be either some number of
receive transition, properly ordered, or a single send transition. In fact, Erlang's semantics is
such that transitions cannot be mixed in other ways and should have at
most one send transition from any state, a condition we refer to as
\emph{affine sends}. We use this condition in order to reason about the possibility
of non-determinism in a trace.

In the following, we will assume that
\emph{self-messaging} is disallowed: i.e., given any local state $(s_i, m_i, \gamma_i)$ for $s_i \in S_i$, $\delta(s_i, P!t)$ is undefined for any $P$ such that $\gamma_i(P) = i$\footnote{This restriction is given to avoid both static messages to self, such as in $1!t$, and dynamic ones, such as $P!t$ where $P$ is bound to $1$. A more strict approach would be to require that each process only sends static messages to each other.}.
For a class of binary models where no self messaging is allowed we observe (and formally
prove below) that there will be no race conditions. A pre-requisite of
a race condition (Definition~\ref{def:race-condition}) is that an
actor can receive two incoming messages at the same trace step. By the
stated conditions, there can only be one sender of messages at a time
and therefore there is at most one transition for every state of an
actor, across all global configurations. Thus, if a final state is
reached then there will only be one possible trace to it.

To prove that we can have at most one transition at a specific time for an actor, we need to look at the mailbox. Send actions are affine (i.e., deterministic) and so we only need to show that receive transitions are deterministic. The proof considers two possible scenarios where we have two different global configurations for the system:
\begin{itemize}
    \item Automata in both configurations are in states about to send, or both about to receive. Since both automata are going to perform the same type of action, they will only affect one mailbox; the mailbox configuration cannot diverge in this case.
    \item One automata is in a state about to send and the other to receive. If the actor in the receive state does not have a valid message to consume, the only valid action is the send (i.e., the other configuration progresses), otherwise both actions would impact the same mailbox making these two configurations diverge. However the send state is a ``constructive'' action, which will append a message at the end of the mailbox, while the receive action is ``destructive'' consuming a message from the mailbox, their disjoint nature will result in the convergence of the mailbox.
\end{itemize}

\begin{prop}[Convergence]
Let $X_1 = (S_1, o_1, F_1, \mathcal{L}_{!_1}, \mathcal{L}_{?_1}, \delta_1)$ and $X_2 = (S_2, o_2, F_2, \mathcal{L}_{!_2}, \mathcal{L}_{?_2}, \delta_2)$ form a protocol of two actors, where there are no mixed transitions, only affine sends, and no self messaging.
Such a protocol, if it converges to a final state, will do so deterministically, i.e.,
any global configuration with converge to the same final global configuration, and as a consequence, be race free.
\end{prop}

\begin{proof}
We prove our desired goal by showing that any two traces $T_1$ and $T_2$ of the same size are the same. We note that a possible sequence of states must follow a pattern $\langle s_0,...,s_n \rangle$, for some $n$, and show that for any $k \leq n$, the prefix of size $k$ of $T_1$ and $T_2$ is the same. Taking $k$ to be $n$, they are the same. This proof follows by well-founded induction on $k$ (which has an upper bound):

\begin{enumerate}
    \item Base case: we take $k = 1$. Both prefixes should then be $\langle s\rangle$, where $s$ is the initial global state. This follows by definition of a trace, as it deterministically specifies which is the first state
    for any given configuration.
    \item Inductive step: our inductive assumption says that our prefixes $\langle s_0...s_k\rangle$ match. If $k = n$, then we are done. If, however, $k \leq n$, then $s_k$ is not a final state and we have that $T_1 = \langle s_0...s_k t_1\rangle$ and $T_2 = \langle s_0...s_k t_2\rangle$. We must now show that $t_1 = t_2$.
    Since we can only have one other process sending messages, we have that $\forall j \in [1,k+1],  \vert I^{j + 1}_i \vert \leq 1$. Now, since the size of the sets is at most one, consuming a message would nullify the set, therefore we can represents the mailbox as a subsequence of the set $\{I^0_i,I^1_i,..,I^k_i\}$, with $m_i = \langle{I^{x_1}_i,I^{x_2}_i,...I^{x_y}_i}\rangle$ where $x_1 < x_2 < ... < x_y < k + 1$. We find the greatest $z$ such that $\forall j \in [1,z], \vert \delta(s_k, ?I^j_i) \vert = 0$. Let $P_i$ be the set of all possible states after a receive transition. If $j = y$, then $\vert P_i \vert = 0$, so no messages can be consumed and we stay put at the same state, otherwise $P_i = \{s' \vert \forall m \in I^{z + 1}_i, \delta(s_k, ?m) = s'\}$. Since we can receive at most one message, $\vert I^{z + 1}_i \vert \leq 1$ therefore $\vert P_i \vert \leq 1$, having at most one possible transition, so just one possible new state. This can't be zero as we have $t_1$ and $t_2$: thus the list of possible next states has only one member, and $t_1 = t_2$ as expected.
\end{enumerate}
\end{proof}
%\mnote{This is induction on just one actor, it should be symmetric, but it can also be proved for both actors or N actors.}

\paragraph{Compatibility in Erlang}
Early work on CFSMs identified the notion of \emph{compatibility}
between machines as a key step towards guaranteeing progress of
systems~\cite{gouda1984progress}. Compatibility is the automata
analogue of what is commonly known as duality: that every receive has a
corresponding send and vice versa. A pair of compatible, deterministic
machines is then free from deadlock and unspecified
receptions~\cite{gouda1984progress}.

Due to Erlang's design principles, we might not always care about ending up
with an empty mailbox and admit systems as being compatible even if some messages
are not received, leaving remaining messages in the mailbox. For future
work we thus propose `tiers' of final conditions on a system which
characterise, roughly, different levels of compatibility.

In the following, let the set of all possible traces for a system of CAAs be $Y$.

\begin{definition} (Tier 1) A system is \emph{strongly compatible} iff, $ \forall T \in Y , t_{|T|} = \langle{(f_i, \sequEmpty, \gamma_i)}\rangle^N_{i = 1}$ where $f_i \in F_i$, i.e., no process has any message left in their mailboxes and they have reached final states.

\end{definition}

\begin{definition} (Tier 2)
A system is \emph{weakly compatible} iff, $ \forall T \in Y , t_{|T|} = \langle{(f_i, m_i, \gamma_i)}\rangle^N_{i = 1}$ where $f_i \in F_i$, i.e., some processes may have some messages left in their mailbox, but all have reached final states.
\end{definition}

\begin{definition} (Tier 3)
A system is \emph{communication-lacking} iff, $ \forall T \in Y , t_{|T|} = \langle{(s_i, \sequEmpty, \gamma_i)}\rangle^N_{i = 1}$ where $s_i \in {S_i/F_i}$, i.e., processes didn't reach their final states but can't continue because of lack of input.
\end{definition}

\begin{definition} (Tier 4) If none of the previous conditions are met, the system is said to be \emph{incompatible}.

\end{definition}

We remark that the system in Example~\ref{exm:trace} is strongly compatible, as all
possible traces will end up with empty mailboxes and in final states.

\section{Discussion and Related work}
\label{sec:discussion}
Fowler describes a framework for generating runtime monitors for
Erlang/OTP's \texttt{gen\_server} behaviours from multiparty session
types as conceived of in the Scribble
language~\cite{fowler2016erlang}. This leverages the idea, due to
Deni{\'e}lou and Yoshida~\cite{denielou2012multiparty} of projecting
Scribble's global types (multiparty session types) into local types,
and then implementing local types as CFSMs. It is not clear
however to what extent this models Erlang's general mailbox semantics.
This warrants a further investigation.

A classic mantra of Erlang is to ``let it crash'' (or ``let it
fail'').  Our model here does not deal with process failure, although
recent models have incorporated such
aspects~\cite{DBLP:conf/coordination/BocchiLTV22}. In the
model of Bocchi et al.~\cite{DBLP:conf/coordination/BocchiLTV22}, actors communicate via unidirectional links to
their mailboxes, similar to the structure of CFSMs, but with increased
flexibility in the way that steps occur. The approach doesn't integrate
pattern matching or dynamic topologies. Mailbox MSCs (Message
Sequence Charts)~\cite{bollig2021unifying,bouajjani2018completeness},
build a message sequent chart model of processes with a single
incoming channel and matching semantics similar in philosophy to our
model but not directly based in the CFSM tradition. A deeper comparison with our approach is further work. One considerable difference in our
approach is the integration of dynamic topologies by the `memory'
environment $\gamma$ for each process, which enables messages to be sent
to a variable which is a process identifier bound by a preceding receive.

In preliminary work, we have created a tool for extracting a CAA model from
the Erlang code, and also the other way around, generating an Erlang skeleton of
concurrent communicating code from a description of a CAA protocol. Further
work includes developing this into a tool for analysis and specification of Erlang
programs. We also note that due to the possibility to describe mobile processes, we
conjecture that Milner's CPS translations from the $\lambda$-calculus into
the $\pi$-calculus could be adapted to use CAAs as a target language. If that's the
case, then it follows that CAAs describe a Turing-complete model of computation, and we intend
to investigate this possibility.

\paragraph{Acknowledgments}

With thanks to Kartik Jalal for insights coming from model extraction of
CAAs from Erlang, Simon Thompson for his Erlang expertise in the early
days of this idea, and University of Kent undergraduates taking \emph{CO545:
Functional and Concurrent Programming} between 2017-20 for being a
test audience for this model and its use in diagnosing race
conditions and deadlocks. This work was partly supported by EPSRC project
EP/T013516/1 (Granule).  Orchard is also supported in part by the
generosity of Eric and Wendy Schmidt by recommendation of the Schmidt
Futures program.

\nocite{*}
\bibliographystyle{eptcs}
\bibliography{references}

\end{document}